\documentclass[a4paper]{article}
\usepackage{lmodern}
\pdfoutput=1
\usepackage{graphicx,wrapfig}
\usepackage{lipsum}
\usepackage{amsmath,amssymb,mathrsfs}
\usepackage{hyperref}
\usepackage{url}
\usepackage{mathtools}
\usepackage{changepage}
\usepackage[ruled,vlined]{algorithm2e}
\usepackage{amsmath}
\usepackage{multirow}
\usepackage{wrapfig}
\usepackage{subfig}
\usepackage{color}
\usepackage{epsfig,enumerate,amsmath,amsfonts,amssymb,amsthm,mathrsfs,ifpdf}
\usepackage{indentfirst,relsize}
\usepackage{setspace,graphicx}
\usepackage{latexsym}
\usepackage[all]{xy}
\usepackage[dvipsnames]{pstricks}
\usepackage{pst-grad} % For gradients
\usepackage{pst-plot} % For axes
\usepackage[margin = 2.50cm]{geometry}

\usepackage{authblk} % to add author affiliations
\usepackage{algpseudocode}
\usepackage{color,soul}
\usepackage{enumitem}% http://ctan.org/pkg/enumitem

\usepackage{orcidlink}% for orcid id
\usepackage{tabularx}
\usepackage{booktabs,tabularx,array}
\usepackage{tikz}
\usetikzlibrary{positioning}
\usetikzlibrary{positioning,arrows.meta}

\usepackage[T1]{fontenc}

\newcommand{\remove}[1]{}

\newtheorem{theorem}{Theorem}%[section]
\newtheorem{lemma}[theorem]{Lemma}

\newtheorem{claim}[theorem]{Claim}

\newtheorem{observation}[theorem]{Observation}

\newcommand{\MEOP}{\textnormal{\textsc{Maximum Edge Open Packing}}}
\newcommand{\MIS}{\textnormal{\textsc{Maximum Independent Set}}}

\newcommand{\rhooe}{\rho^{\mathrm{o}}_{\mathrm{e}}}

\usepackage{todonotes}
\usepackage[most]{tcolorbox}

\title{Maximum Edge Open Packing on AT-Free, Chordal, and Convex Bipartite Graphs}
\author[1]{Gautam K. Das\footnote{gkd@iitg.ac.in}}
\author[1]{Kamal Santra \orcidlink{0009-0006-5997-1452} \footnote{kamal.7.2013@gmail.com, kamal.santra@iitg.ac.in}}
\affil[1]{Department of Mathematics\\
	
	Indian Institute of Technology Guwahati\\
	
	Guwahati, 781039, Assam, India}

\date{}

\begin{document}

	\maketitle
\begin{abstract}
An edge open packing is a set of edges whose endpoints induce a disjoint union of stars. We study the problem of finding a maximum edge open packing in AT-free, chordal, and convex bipartite graphs. For AT-free graphs, we use the oriented star-conflict graph $A_G$, introduced by Das and Santra, for which $\rho_e^o(G)=\alpha(A_G)$, and prove that $A_G$ is AT-free whenever $G$ is AT-free. Consequently, a known maximum independent set algorithm for AT-free graphs yields an $O(n^2+m^4)$-time algorithm for Maximum Edge Open Packing, where $n$ and $m$ denote the numbers of vertices and edges of $G$, respectively. For chordal graphs, we develop an $O(n^4)$-time dynamic programming algorithm over a nice tree decomposition derived from a clique tree, exploiting the structural fact that the endpoint set of an edge open packing intersects every clique in at most two vertices. Finally, for convex bipartite graphs, we obtain an $O(n^4)$-time dynamic programming algorithm based on two boundary indices that separate consecutive star components. This result extends the previously known algorithm for biconvex bipartite graphs and, in particular, provides an improved explicit running-time bound for that class.
	
\end{abstract}

\section{Introduction}

An induced matching selects edges whose endpoints induce a matching. It is
a standard model for choosing mutually isolated interactions and has been
studied extensively in graph algorithms~\cite{Cameron1989,Lozin2002,DabrowskiEtAl2013}.
Edge open packing, introduced by Chelladurai et al.~\cite{ChelladuraiEtAl2022},
allows several selected edges to share a common centre while retaining
separation between different components. For example, all $t$ edges of a
star $K_{1,t}$ form an edge open packing, whereas an induced matching
contains at most one of them.

One motivation comes from channel assignment under a secondary-interference
model in packet radio networks~\cite{ChelladuraiEtAl2022,CardosoEtAl2019}.
In this model, two selected links conflict when a third link joins their
endpoints. A maximum edge open packing is therefore a largest feasible set of links
for one channel, and the colour classes of an injective edge colouring are
edge open packings. The model permits several links with a shared star
centre. More abstractly, it describes the selection of mutually isolated
groups of interactions, each organised around one hub.

Formally, let $G=(V,E)$ be a graph. Two distinct edges have a \emph{common
edge} if an edge other than those two joins an endpoint of one to an endpoint
of the other. An \emph{edge open packing} is a set $D\subseteq E$ in which
no two distinct edges have a common edge. Its maximum cardinality is the
\emph{edge open packing number}, denoted by $\rhooe(G)$, and the optimisation
problem of finding such a set is called \MEOP. If $V(D)$ denotes the endpoints
of the edges in $D$, then $D$ is an edge open packing precisely when the
connected components of $G[V(D)]$ are stars, including single edges.
Thus, the problem extends induced matching by permitting larger star
components.

Bre\v{s}ar and Samadi~\cite{BresarSamadi2024} proved that the decision version
is NP-complete on restricted classes including planar graphs of maximum
degree four and Eulerian bipartite graphs, and gave a linear-time algorithm
for trees. Structural work relates edge open packing to induced matching,
diameter, minimum degree, and clique number~\cite{BresarEtAl2025,ChelladuraiEtAl2022,PandeySantra2025}.
Santra~\cite{Santra2025Chordal} obtained polynomial-time algorithms for proper
interval, block, and split graphs. Bhyravarapu et al.~\cite{BhyravarapuEtAl2026} gave polynomial-time algorithms for
distance-hereditary and biconvex bipartite graphs, with an explicit
$O(n^{42})$ bound for the latter class. They also obtained a
$2^{O(\omega)}\operatorname{poly}(n)$-time algorithm on chordal graphs,
where $\omega$ is the clique number.

 Das and Santra~\cite{DasSantra2026} introduced the \emph{oriented star-conflict graph} $A_G$ and established the identity $\rhooe(G)=\alpha(A_G)$. The vertices of $A_G$ correspond to oriented edges of $G$, where the first endpoint is designated as the centre of a star and the second as one of its leaves; the precise adjacency rule is recalled in Section~\ref{sec:atfree}. Using this construction, they obtained $O(n^2+m^2)$-time algorithms for permutation and interval graphs, where $n=|V(G)|$ and $m=|E(G)|$; in terms of $n$ alone, this gives an $O(n^4)$ bound~\cite{DasSantra2026}. They treated well-partitioned chordal graphs separately.

More recently, Das and Santra~\cite{DasSantraCactusCocomp2026} used the same oriented star-conflict graph to study \textsc{Maximum Edge Open Packing} on cocomparability graphs. They proved that if $G$ is a cocomparability graph, then $A_G$ is also a cocomparability graph. Since maximum independent set is polynomial-time solvable on cocomparability graphs, this yields an $O(n^2+m^2)$-time algorithm, and hence an $O(n^4)$-time algorithm in terms of $n$ alone~\cite{DasSantraCactusCocomp2026}. This result simultaneously generalizes the earlier algorithms for permutation and interval graphs, since both classes are contained in the class of cocomparability graphs.

The present work takes this structural approach one step further. We prove that the transformation $G\mapsto A_G$ preserves AT-freeness: if $G$ is AT-free, then so is $A_G$. Since cocomparability graphs form a subclass of AT-free graphs, this extends the oriented star-conflict framework to a strictly broader setting. The graph $A_G$ has $2m$ vertices and can be constructed in $O(n^2+m^2)$ time. Applying the $O(N^4)$-time maximum independent set algorithm of Broersma et al.~\cite{BroersmaEtAl1999} to $A_G$, where $N=|V(A_G)|=2m$, gives an overall running time of $O(n^2+m^4)$, which is $O(n^8)$ when expressed only in terms of $n$. The progression of these results and the corresponding graph-class inclusions are illustrated in Figure~\ref{fig:eop-class-hierarchy}.

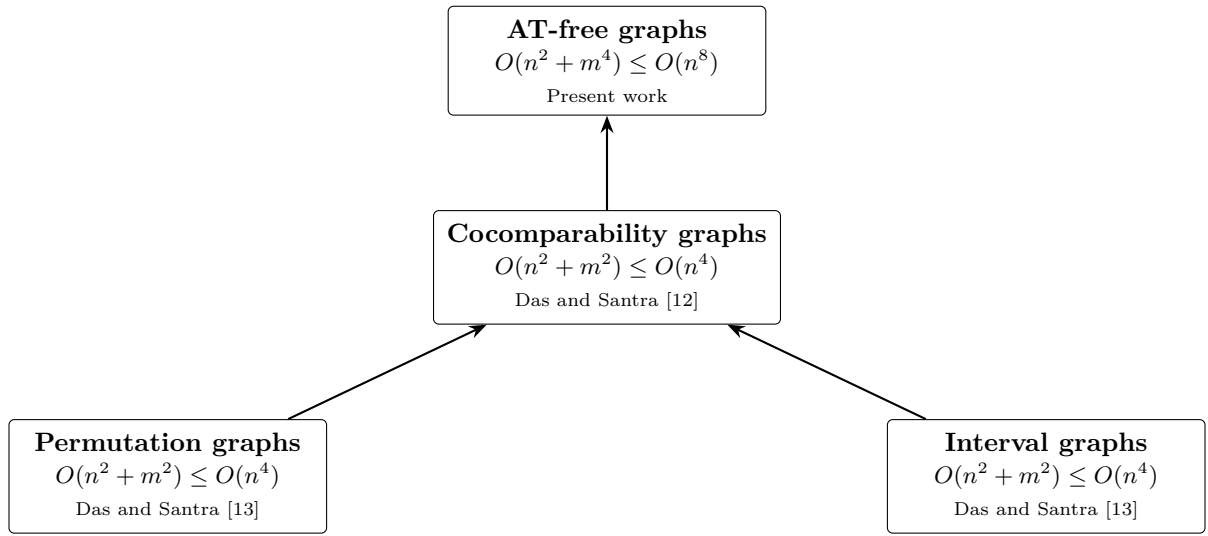
\begin{figure}[htbp]
	\centering
	\begin{tikzpicture}[
		class/.style={
			draw,
			rounded corners=2pt,
			align=center,
			minimum width=4.2cm,
			minimum height=1.05cm,
			inner sep=5pt
		},
		edge/.style={->,>=Stealth,thick},
		node distance=1.25cm and 1.4cm
		]
		
		\node[class] (atfree)
		{\textbf{AT-free graphs}\\
			\small $O(n^2+m^4)\leq O(n^8)$\\
			\scriptsize Present work};
		
		\node[class, below=of atfree] (cocomp)
		{\textbf{Cocomparability graphs}\\
			\small $O(n^2+m^2)\leq O(n^4)$\\
			\scriptsize Das and Santra~\cite{DasSantraCactusCocomp2026}};
		
		\node[class, below left=of cocomp] (perm)
		{\textbf{Permutation graphs}\\
			\small $O(n^2+m^2)\leq O(n^4)$\\
			\scriptsize Das and Santra~\cite{DasSantra2026}};
		
		\node[class, below right=of cocomp] (interval)
		{\textbf{Interval graphs}\\
			\small $O(n^2+m^2)\leq O(n^4)$\\
			\scriptsize Das and Santra~\cite{DasSantra2026}};
		
		\draw[edge] (perm) -- (cocomp);
		\draw[edge] (interval) -- (cocomp);
		\draw[edge] (cocomp) -- (atfree);
		
	\end{tikzpicture}
	\caption{Graph-class hierarchy for the oriented star-conflict approach to
		\textsc{Maximum Edge Open Packing}. An arrow from $\mathcal{C}_1$ to
		$\mathcal{C}_2$ indicates that $\mathcal{C}_1\subseteq\mathcal{C}_2$.
		The displayed bounds are the corresponding running times for computing
		a maximum edge open packing.}
	\label{fig:eop-class-hierarchy}
\end{figure}

\medskip
\noindent
\textbf{Our contributions.}
We give polynomial-time algorithms for three graph classes. First, we prove
that $A_G$ is AT-free whenever $G$ is AT-free. Combining this preservation
theorem with the correspondence of Das and Santra and the independent-set
algorithm of Broersma et al.~\cite{BroersmaEtAl1999} yields an
$O(n^2+m^4)$-time algorithm for \MEOP, where $n=|V(G)|$ and $m=|E(G)|$.

Second, we give an $O(n^4)$-time algorithm for chordal graphs, without a
bound on the clique number. We use a standard nice tree decomposition
obtained from a clique tree~\cite{Kloks1994}. The endpoint set of an edge
open packing meets each clique in at most two vertices. Recording these
vertices and their roles in a partially constructed star gives quadratically
many states per bag and a constant number of label combinations for each
active subset at a join.

Third, we give an $O(n^4)$-time algorithm for convex bipartite graphs.
Convexity orders the connected components of a solution from left to right.
Two boundary indices record all restrictions on the next component, allowing
a dynamic program that appends one star at a time. This extends the earlier
biconvex result to a larger class and improves its explicit running-time
bound. These extensions are complementary: bipartite permutation graphs
are AT-free and biconvex, but biconvex graphs need not be AT-free.
Section~\ref{sec:preliminaries} gives the inclusion relations and a small
example separating the two classes.

\section{Preliminaries}
\label{sec:preliminaries}

All graphs are finite, simple, and undirected. For a graph $G$, write
$n=|V(G)|$ and $m=|E(G)|$. The open and closed neighbourhoods of a vertex
$v$ are $N_G(v)$ and $N_G[v]$, respectively. For $U\subseteq V(G)$,
$G[U]$ denotes the subgraph induced by $U$. An independent set consists of
pairwise nonadjacent vertices, and $\alpha(G)$ is its maximum size.
For $D\subseteq E(G)$, let $V(D)$ be the set of endpoints of its edges.
Throughout the paper, a \emph{star} means $K_{1,t}$ with $t\geq1$.
In particular, $K_2$ is a star, and either endpoint may be chosen as its
centre. A star with at least two edges has a unique centre.

We use the edge open packing definition from the introduction. The following
characterisation also specifies exactly which edges are selected.

\begin{observation}
\label{obs:star-characterization}
An edge set $D\subseteq E(G)$ is an edge open packing if and only if every
connected component of $G[V(D)]$ is a star. In this case,
$E(G[V(D)])=D$.
\end{observation}

\begin{proof}
If $D$ is an edge open packing, an edge in $E(G[V(D)])\setminus D$ would
join two distinct selected edges incident with its endpoints. Hence
$E(G[V(D)])=D$. This graph contains neither a triangle nor a path with
three edges as a subgraph: in either case one edge is common to two others.
A connected graph with an edge and neither of these subgraphs is a star.
Conversely, if every component of $G[V(D)]$ is a star, each leaf belongs
to $V(D)$, so its unique incident edge must belong to $D$. Thus all star
edges are selected, and no third edge joins two of them in the forbidden
way.
\end{proof}

A graph is \emph{chordal} if it has no induced cycle of length at least four.
Its maximal cliques admit a \emph{clique tree}: for each vertex $v$, the
maximal cliques containing $v$ induce a connected subtree~\cite{FulkersonGross1965,Golumbic1980}.
An \emph{asteroidal triple} is an independent set $\{u,v,w\}$ such that
each pair is joined by a path avoiding the closed neighbourhood of the
third vertex. A graph is \emph{AT-free} if it contains no asteroidal triple.

A bipartite graph $G=(X\cup Y,E)$ is \emph{convex bipartite} if one part,
say $X$, admits an ordering $x_1,\ldots,x_p$ in which every neighbourhood
$N(y)$, $y\in Y$, is consecutive. After deleting isolated vertices, write
$N(y)=\{x_{\ell(y)},\ldots,x_{r(y)}\}$ with
$1\leq\ell(y)\leq r(y)\leq p$. A convex ordering can be found by testing
the consecutive-ones property of the bipartite adjacency matrix~\cite{BoothLueker1976}.
A graph is \emph{biconvex bipartite} if both parts admit such orderings.
A \emph{bipartite permutation graph} is a bipartite graph whose vertices
can be represented by line segments between two parallel lines, with
adjacency corresponding to intersection. A bipartite graph is
\emph{chordal bipartite} if it has no induced cycle of length at least six.

The standard strict inclusion chain is bipartite permutation graphs
$\subsetneq$ biconvex bipartite graphs $\subsetneq$ convex bipartite graphs
$\subsetneq$ chordal bipartite graphs~\cite{BrandstadtLeSpinrad1999}.
Bipartite permutation graphs are AT-free, but biconvex graphs need not be.
For example, take a cycle $a b c d a$ and add pendant vertices
$a',b',c',d'$ at the corresponding cycle vertices. The two part orders
$(b',a,c,d')$ and $(a',b,d,c')$ show that this graph is biconvex,
whereas $\{a',b',c'\}$ is an asteroidal triple. Consequently, our AT-free
result does not subsume the earlier biconvex result. A claw with each edge
subdivided once is convex but not biconvex: order its three subdivision
vertices on one side to obtain convexity; on the other side, making each
pair consisting of the original centre and an original leaf consecutive
would require the centre to be consecutive with three different vertices.

\section{Maximum Edge Open Packing in AT-Free Graphs}
\label{sec:atfree}

We use the oriented star-conflict graph of Das and
Santra~\cite{DasSantra2026} to reduce \MEOP\ on an AT-free graph to \MIS\
on another AT-free graph.

\subsection{The oriented star-conflict graph}

We recall the construction from~\cite{DasSantra2026}. An \emph{oriented
edge} of $G$ is an ordered pair $(c,\ell)$ such that $c\ell\in E(G)$.
The \emph{oriented star-conflict graph} $A_G$ has all oriented edges of $G$
as its vertices. Two distinct oriented edges $(c,\ell)$ and $(d,r)$ are
\emph{compatible} if one of the following conditions holds.

\begin{enumerate}[label=(\roman*)]
	\item They have the same centre: $c=d$, $\ell\neq r$, and
	$\ell r\notin E(G)$.
	\item The vertices $c,\ell,d,r$ are all distinct, and none of the four
	pairs $cd,cr,\ell d,\ell r$ is an edge of $G$.
\end{enumerate}

Two vertices of $A_G$ are adjacent if and only if they are incompatible.
In particular, $(c,\ell)$ and $(\ell,c)$ are adjacent, so an independent set
cannot select both orientations of one edge. The first coordinate specifies
the centre of a selected star and the second specifies a leaf; for a
single-edge component, either endpoint may serve as its centre.

\begin{lemma}[\cite{DasSantra2026}]
	\label{lem:eop-independent}
	For every graph $G$, $\rhooe(G)=\alpha(A_G)$.
\end{lemma}

\begin{proof}
	For completeness, let $D$ be an edge open packing. By
	Observation~\ref{obs:star-characterization}, the connected components of
	$G[V(D)]$ are stars. Orient their edges from centre to leaf, choosing either
	centre for each $K_2$. The resulting $|D|$ oriented edges satisfy the two
	compatibility rules and form an independent set in $A_G$.
	Conversely, an independent set $I$ contains at most one orientation of each
	edge. Its members with a common first coordinate describe a star with
	pairwise nonadjacent leaves. Those with different first coordinates have
	disjoint endpoint sets with no edge between them. Thus their $|I|$
	underlying edges form an edge open packing, proving the equality.
\end{proof}

For example, let $G=C_5$ with vertices $v_0,\ldots,v_4$ in cyclic order,
and read subscripts modulo five. Write $a_i=(v_i,v_{i-1})$ and
$b_i=(v_i,v_{i+1})$. The pair $a_i,b_i$ is compatible for each $i$.
Oriented edges with different centres are incompatible, since $C_5$ has
no induced $2K_2$. Hence $\overline{A_G}=5K_2$, or equivalently
$A_G=K_{2,2,2,2,2}$, with parts $\{a_i,b_i\}$.
In particular, $\rhooe(C_5)=\alpha(A_{C_5})=2$.
Figure~\ref{fig:cycle-conflict} shows the complement to display the
compatibility pairs without drawing all forty conflict edges.

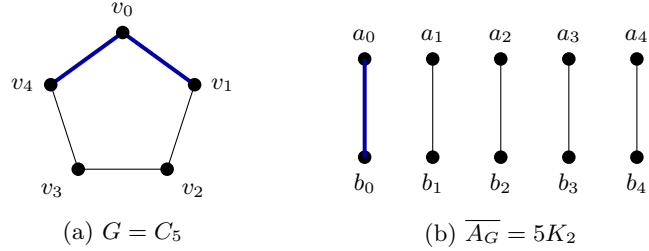
\begin{figure}[htb]
\centering
\begin{tikzpicture}[every node/.style={font=\small}]
  \coordinate (v0) at (90:1);
  \coordinate (v1) at (18:1);
  \coordinate (v2) at (-54:1);
  \coordinate (v3) at (-126:1);
  \coordinate (v4) at (162:1);
  \draw (v0)--(v1)--(v2)--(v3)--(v4)--(v0);
  \draw[line width=1.5pt,blue!65!black] (v4)--(v0)--(v1);
  \node[circle,fill,inner sep=1.8pt,label=above:$v_0$] at (v0) {};
  \node[circle,fill,inner sep=1.8pt,label=right:$v_1$] at (v1) {};
  \node[circle,fill,inner sep=1.8pt,label=below right:$v_2$] at (v2) {};
  \node[circle,fill,inner sep=1.8pt,label=below left:$v_3$] at (v3) {};
  \node[circle,fill,inner sep=1.8pt,label=left:$v_4$] at (v4) {};
  \node at (0,-1.65) {(a) $G=C_5$};
  \foreach \i in {0,...,4} {
    \coordinate (a\i) at ({3.2+0.9*\i},0.65);
    \coordinate (b\i) at ({3.2+0.9*\i},-0.65);
    \draw (a\i)--(b\i);
    \node[circle,fill,inner sep=1.8pt,label=above:$a_{\i}$] at (a\i) {};
    \node[circle,fill,inner sep=1.8pt,label=below:$b_{\i}$] at (b\i) {};
  }
  \draw[line width=1.5pt,blue!65!black] (a0)--(b0);
  \node at (5,-1.65) {(b) $\overline{A_G}=5K_2$};
\end{tikzpicture}
\caption{The cycle $C_5$ and the complement of its oriented star-conflict
graph. Each edge on the right joins compatible orientations; all other
pairs are adjacent in $A_G$. The highlighted pair $a_0,b_0$ represents
the highlighted maximum edge open packing on the left.}
\label{fig:cycle-conflict}
\end{figure}

\subsection{Preservation of AT-freeness}

\begin{theorem}
	\label{thm:at-preservation}
	If $G$ is AT-free, then $A_G$ is AT-free.
\end{theorem}

\begin{proof}
	We prove the contrapositive. Suppose that $A_G$ contains an asteroidal triple
	$\{e_1,e_2,e_3\}$, where $e_i=(c_i,\ell_i)$ for $i\in\{1,2,3\}$. Since the
	three vertices are pairwise nonadjacent in $A_G$, the corresponding oriented
	edges are pairwise compatible. Therefore, $\ell_1,\ell_2,\ell_3$ are distinct
	and pairwise nonadjacent in $G$.
	
	We show that $\{\ell_1,\ell_2,\ell_3\}$ is an asteroidal triple in $G$. Fix
	distinct $i,j,k\in\{1,2,3\}$. Since $\{e_1,e_2,e_3\}$ is an asteroidal triple
	in $A_G$, there is an $e_i$--$e_j$ path in $A_G$ avoiding
	$N_{A_G}[e_k]$. Write this path as
	$P=e_i=f_0,f_1,\ldots,f_t=e_j$, where $f_s=(a_s,b_s)$, and write
	$e_k=(c_k,\ell_k)$. For $0\leq s\leq t$, put $R_s=\{b_s\}$ if
	$a_s=c_k$, and $R_s=\{a_s,b_s\}$ otherwise.
	
	Every vertex $f_s$ of $P$ is compatible with $e_k$. If $a_s=c_k$, then
	compatibility of $(c_k,b_s)$ and $(c_k,\ell_k)$ gives
	$b_s\neq\ell_k$ and $b_s\ell_k\notin E(G)$. If $a_s\neq c_k$, compatibility
	gives $a_s,b_s\notin N_G[\ell_k]$. Hence,
	$R_s\cap N_G[\ell_k]=\varnothing$ for every $s$.
	
	We next claim that consecutive sets $R_s$ and $R_{s+1}$ touch in $G$; that
	is, they intersect or an edge of $G$ has one endpoint in each set. Suppose
	otherwise. Write $f_s=(a,b)$ and $f_{s+1}=(d,r)$. If $a=d=c_k$, then
	$R_s=\{b\}$ and $R_{s+1}=\{r\}$; the assumption gives $b\neq r$ and
	$br\notin E(G)$, so the two oriented edges are compatible by the same-centre
	rule. If $a=d\neq c_k$, then both sets contain $a$, contrary to their being
	disjoint.
	
	Now suppose that $a\neq d$. If neither $a$ nor $d$ equals $c_k$, then
	$R_s=\{a,b\}$ and $R_{s+1}=\{d,r\}$. Their disjointness and the absence of an
	edge between them show that $a,b,d,r$ are distinct and that all four cross
	edges are absent. Thus, $f_s$ and $f_{s+1}$ are compatible. If $a=c_k$, then
	$R_s=\{b\}$ and $R_{s+1}=\{d,r\}$. The assumption gives
	$b\neq d,r$ and $bd,br\notin E(G)$. Compatibility of $(d,r)$ with
	$(c_k,\ell_k)$ gives $c_k\neq d,r$ and $c_kd,c_kr\notin E(G)$. Again the
	two oriented edges are compatible. The case $d=c_k$ is symmetric.
	
	In every case, $f_s$ and $f_{s+1}$ would be compatible, contradicting their
	adjacency on $P$. Hence consecutive sets touch. Each $R_s$ induces a connected
	subgraph, and therefore $G[\bigcup_{s=0}^{t}R_s]$ is connected. It contains
	$\ell_i$ and $\ell_j$ and avoids $N_G[\ell_k]$. Thus, there is an
	$\ell_i$--$\ell_j$ path avoiding $N_G[\ell_k]$.
	
	The indices $i,j,k$ were arbitrary. Therefore,
	$\{\ell_1,\ell_2,\ell_3\}$ is an asteroidal triple in $G$. Hence, if $A_G$
	contains an asteroidal triple, then so does $G$, proving the contrapositive.
\end{proof}

\subsection{Algorithm and running time}

Given an AT-free graph $G$, construct $A_G$. By
Theorem~\ref{thm:at-preservation}, $A_G$ is AT-free, and by
Lemma~\ref{lem:eop-independent}, $\rhooe(G)=\alpha(A_G)$. We compute a maximum
independent set of $A_G$ using the algorithm of Broersma et
al.~\cite{BroersmaEtAl1999} and translate it back into an edge open packing of
$G$ as in the proof of Lemma~\ref{lem:eop-independent}.

\begin{theorem}
	\label{thm:atfree-algorithm}
	The \MEOP\ problem can be solved in $O(n^2+m^4)$ time on an AT-free graph with
	$n$ vertices and $m$ edges. In particular, the running time is $O(n^8)$ when
	expressed only in terms of $n$.
\end{theorem}

\begin{proof}
	The graph $A_G$ has one vertex for each oriented edge of $G$, so
	$|V(A_G)|=2m$. An adjacency matrix of $G$ can be constructed in $O(n^2)$
	time, after which compatibility of two oriented edges can be tested in constant
	time. Thus, $A_G$ can be constructed in $O(n^2+m^2)$ time.
	
	The algorithm of Broersma et al.~\cite{BroersmaEtAl1999} solves \MIS\ on an
	$N$-vertex AT-free graph in $O(N^4)$ time. Applying it with $N=2m$ takes
	$O(m^4)$ time. Translating the independent set back into an edge open packing
	takes $O(m)$ time. The total running time is therefore $O(n^2+m^4)$, which is
	$O(n^8)$ because $m=O(n^2)$.
\end{proof}

	\section{Maximum Edge Open Packing in Chordal Graphs}
\label{sec:chordal}

In this section, we give a polynomial-time algorithm for chordal graphs. Let
$G=(V,E)$ be a chordal graph on $n$ vertices. If $G$ is disconnected, we solve
the problem independently on its connected components and add the optimum
values. We therefore assume throughout the description that $G$ is connected.

Let $\mathcal{T}$ be a clique tree of $G$. From $\mathcal{T}$, we construct a
nice rooted tree decomposition $(T,\mathcal{B})$ whose root and leaves have
empty bags and whose nodes are leaf, introduce, forget, or join nodes. Along
an edge between two maximal cliques, vertices are first forgotten down to the
intersection and then introduced up to the other clique. Join nodes use copies
of one maximal clique. Thus, every bag of the nice decomposition is a subset
of a maximal clique and is therefore a clique. The standard transformation
has $O(n^2)$ nodes~\cite{Kloks1994}.

For $t\in V(T)$, let $B_t$ be its bag, let $T_t$ be the subtree rooted at $t$,
and let $G_t$ be the subgraph of $G$ induced by all vertices occurring in bags
of $T_t$. We use the following separation properties. If $t$ is an introduce
node with child $t'$ and $B_t=B_{t'}\cup\{v\}$, then $v$ has no neighbour in
$V(G_{t'})\setminus B_{t'}$. If $t$ is a forget node with child $t'$ and
$B_t=B_{t'}\setminus\{v\}$, then $v$ has no neighbour outside $V(G_t)$. If
$t$ is a join node with children $t_1,t_2$, then
$V(G_{t_1})\cap V(G_{t_2})=B_t$, and there is no edge between
$V(G_{t_1})\setminus B_t$ and
$V(G_{t_2})\setminus B_t$.

\begin{observation}
	\label{obs:clique-two}
	\label{obs:forced-edge}
	Let $D$ be an edge open packing of $G$. Every clique of $G$ contains at most
	two vertices of $V(D)$, and every edge of $G[V(D)]$ belongs to $D$.
\end{observation}

Both assertions follow from Observation~\ref{obs:star-characterization}:
$G[V(D)]$ is a star forest, so it has no triangle, and all its edges are
selected.

\subsection{States}

At every node $t$, a state records the vertices of $B_t$ that are intended to
be endpoints of selected edges in the final packing, together with their roles
in the partial induced-star forest. We call these vertices \emph{active}. An
active vertex is either already incident with a selected edge in the processed
part or is pending and must become incident with a selected edge later. Since
$B_t$ is a clique, Observation~\ref{obs:clique-two} implies that a state has at
most two active vertices.

The possible states are the following.

\begin{itemize}
	\item The empty state $\varnothing$ means that no vertex of $B_t$ is active.
	
	\item A singleton state $(x,\lambda)$ has $x\in B_t$ and
	$\lambda\in\{P,C,L\}$. Label $P$ means that $x$ is pending and is not yet
	incident with a selected edge in the processed part. Label $C$ means that
	$x$ is the centre of a partial star with at least one edge. Label $L$ means
	that $x$ is a leaf of a star whose centre lies outside the bag and has at
	least two selected neighbours. In particular, a $K_2$ meeting the bag only
	at $x$ always has label $C$, irrespective of any earlier orientation.
	
	\item A pair state $(\{x,y\},\lambda)$ has distinct $x,y\in B_t$ and
	$\lambda\in\{E,S_x,S_y\}$. Since $B_t$ is a clique, $xy\in E(G)$. In every
	pair state, $xy$ is selected but has not yet been counted. Label $E$ means
	that $xy$ currently forms a $K_2$-component with no counted attachment at
	either endpoint. Label $S_x$ means that the partial star is centred at $x$
	with $y$ as a leaf and at least one counted edge from $x$ to a leaf outside
	the bag. Label $S_y$ is defined symmetrically.
\end{itemize}

A partial solution at $t$ is a set $U\subseteq V(G_t)$ such that
$H=G_t[U]$ is a disjoint union of stars with at least one edge, with the
possible addition of one isolated vertex in $B_t$. The active vertices are
exactly $U\cap B_t$, and every edge of $H$ is selected. The component meeting
the bag, if any, has the canonical label defined above; an isolated active
vertex has label $P$. There is at most one such component because the bag is
a clique. All components disjoint from the bag are complete.

For a state $\sigma$, the table entry $\operatorname{DP}[t,\sigma]$ is the
maximum number of edges of $H$ with at least one endpoint outside $B_t$
over partial solutions with state $\sigma$, or $-\infty$ if none exists.
Thus, the counting invariant is that every selected edge outside the bag
has been counted, while no edge entirely inside the bag has been counted.
The latter case occurs only for the edge $xy$ in a pair state; it is counted
when the first of $x,y$ is forgotten. Canonical labels depend on the current
partial star, so the two endpoints of a $K_2$ need not be assigned permanent
roles.

\subsection{Transitions}

All updates are maximum updates. Entries not produced by a listed transition
remain equal to $-\infty$.

\medskip
\noindent
\textbf{Leaf node.}
For a leaf $t$ with empty bag, set
$\operatorname{DP}[t,\varnothing]=0$.

\medskip
\noindent
\textbf{Introduce node.}
Let $t$ be an introduce node with child $t'$ and
$B_t=B_{t'}\cup\{v\}$. The vertex $v$ may be left inactive, so every child
state is copied to the same state at $t$ with the same value.
If $v$ becomes active, the possible updates are
\begin{align*}
	\operatorname{DP}[t,(v,P)]
	&\gets \max\{\operatorname{DP}[t,(v,P)],
	\operatorname{DP}[t',\varnothing]\},\\
	\operatorname{DP}[t,(\{x,v\},E)]
	&\gets \max\{\operatorname{DP}[t,(\{x,v\},E)],
	\operatorname{DP}[t',(x,P)]\},\\
	\operatorname{DP}[t,(\{x,v\},S_x)]
	&\gets \max\{\operatorname{DP}[t,(\{x,v\},S_x)],
	\operatorname{DP}[t',(x,C)]\},
\end{align*}
for every $x\in B_{t'}$ for which the child entry is finite. No other
transition that makes $v$ active is allowed. Activating $v$ from $(x,L)$
would give the leaf $x$ a second selected neighbour, while activating it from
a pair state would create three active vertices in a clique. No selected edge
is counted at an introduce node.

\medskip
\noindent
\textbf{Forget node.}
Let $t$ be a forget node with child $t'$ and
$B_t=B_{t'}\setminus\{v\}$. If $v$ does not occur in a child state $\sigma$,
copy the state unchanged. The state $(v,P)$ is rejected because $v$ would be
forgotten before becoming incident with a selected edge. For
$\lambda\in\{C,L\}$, the state $(v,\lambda)$ produces the empty state
at $t$ with the same value.

Now suppose that the child state is a pair state involving $v$ and another
vertex $x\in B_t$. The selected edge $vx$ is counted at this forget step. The
updates are
\begin{align*}
	\operatorname{DP}[t,(x,C)]
	&\gets
	\max\{\operatorname{DP}[t,(x,C)],
	\operatorname{DP}[t',(\{v,x\},E)]+1\},\\
	\operatorname{DP}[t,(x,L)]
	&\gets
	\max\{\operatorname{DP}[t,(x,L)],
	\operatorname{DP}[t',(\{v,x\},S_v)]+1\},\\
	\operatorname{DP}[t,(x,C)]
	&\gets
	\max\{\operatorname{DP}[t,(x,C)],
	\operatorname{DP}[t',(\{v,x\},S_x)]+1\}.
\end{align*}
The first update views the counted $K_2$ as a star centred at its remaining
endpoint. In the second update, $v$ is the forgotten centre and $x$ remains a
leaf; the positive attachment in $S_v$ ensures that $v$ has at least two
selected neighbours. In the third, $x$ remains the centre.

\medskip
\noindent
\textbf{Join node.}
Let $t$ be a join node with children $t_1,t_2$, where
$B_t=B_{t_1}=B_{t_2}$. Two child states are combined only when they activate
the same subset of $B_t$.

Define a partial merge operation $\star$ on labels. For the empty state,
$\varnothing\star\varnothing=\varnothing$. For singleton states on the same
vertex, the valid merges are
\[
\begin{array}{c@{\qquad}c@{\qquad}c}
	P\star P=P, & P\star C=C, & C\star P=C,\\
	P\star L=L, & L\star P=L, & C\star C=C.
\end{array}
\]
The merges $L\star L$, $C\star L$, and $L\star C$ are invalid. They would
create a path with three edges: a label $L$ already places the active vertex
at a leaf of a star with at least two edges, and the other nonpending child
adds a second neighbour to that leaf.

For pair states on the same two vertices $x,y$, the valid merges are
\[
\begin{array}{c@{\qquad}c@{\qquad}c}
	E\star E=E, & E\star S_x=S_x, & S_x\star E=S_x,\\
	E\star S_y=S_y, & S_y\star E=S_y, &
	S_x\star S_x=S_x,\\
	&&S_y\star S_y=S_y.
\end{array}
\]
The merges $S_x\star S_y$ and $S_y\star S_x$ are invalid because they would
place attachments at both endpoints of $xy$.

The join recurrence is
\[
\operatorname{DP}[t,\tau]
=
\max_{\sigma_1\star\sigma_2=\tau}
\bigl\{\operatorname{DP}[t_1,\sigma_1]
+\operatorname{DP}[t_2,\sigma_2]\bigr\},
\]
where undefined merges are ignored. No edge is counted at a join node. An
edge entirely inside the bag remains uncounted in both children, while every
selected edge with an endpoint outside the bag belongs to exactly one child
subtree and has already been counted there.

\subsection{Correctness}

The root bag is empty, so the only accepting state at the root $r$ is
$\varnothing$. The value returned by the algorithm is
$\operatorname{DP}[r,\varnothing]$.

\begin{lemma}
	\label{lem:chordal-sound}
	Every finite DP entry satisfies the invariant. Consequently, every accepting
	computation at the root represents an edge open packing whose cardinality is
	its DP value.
\end{lemma}

\begin{proof}
	We argue by induction over the nice tree decomposition. The statement is
	immediate at a leaf.
	
	At an introduce node, an inactive introduced vertex does not change the
	partial solution. An active introduced vertex can interact only with active
	vertices already in the current bag, because it has no neighbour in
	$V(G_{t'})\setminus B_{t'}$. The transition from the empty state creates a
	pending vertex; the transition from $(x,P)$ selects $xv$ and creates a
	$K_2$-component; and the transition from $(x,C)$ attaches $v$ as a new leaf
	of the star centred at $x$. The excluded cases are precisely those that would
	give a leaf a second selected neighbour or place three active vertices in one
	clique. Thus, every introduce transition preserves the invariant.
	
	At a forget node, a pending singleton is rejected. If a singleton centre or
	leaf is forgotten, its component is complete; the separation property implies
	that the forgotten vertex has no neighbour outside $V(G_t)$. If $v$ belongs
	to a pair state with $x$, the edge $vx$ is counted once and the remaining role
	of $x$ is recorded correctly. Thus, all forget transitions preserve the
	invariant and count each selected edge exactly once.
	
	At a join node, put $X_i=V(G_{t_i})\setminus B_t$ for $i\in\{1,2\}$. There is
	no edge between $X_1$ and $X_2$, so no hidden common edge can be created
	between selected vertices lying strictly in different child subtrees. All
	interactions occur through the common bag. Requiring equal active subsets and
	using the listed merge table preserves exactly the compatible partial-star
	shapes. In particular, all attachments of a boundary edge lie at the same
	endpoint. No selected edge is counted twice.
	
	At the root, there is no pending vertex and no uncounted boundary edge. The
	selected vertices therefore induce a disjoint union of induced stars. By
	Observation~\ref{obs:star-characterization}, the selected edges form an edge
	open packing, and every selected edge has been counted exactly once.
\end{proof}

\begin{lemma}
	\label{lem:chordal-complete}
	Every edge open packing of $G$ is represented by the dynamic program.
\end{lemma}

\begin{proof}
	Let $D$ be an edge open packing, and set $U_t=V(D)\cap V(G_t)$ at each
	node $t$. Every vertex in $V(G_t)\setminus B_t$ has no neighbour outside
	$V(G_t)$, by the separation property of a tree decomposition. Thus, each
	vertex of $U_t\setminus B_t$ has a selected neighbour in $U_t$. It follows
	that $G_t[U_t]$ is a star forest whose only possible isolated vertex lies
	in the bag. Observation~\ref{obs:clique-two} ensures that at most two
	vertices are active, and two active vertices form a selected edge.
	
	Assign each boundary component its canonical label. In particular, a
	boundary edge $xy$ receives label $E$ exactly when it has no attachment
	in the processed part; otherwise all attachments have the same centre,
	giving $S_x$ or $S_y$. If a singleton boundary vertex $x$ lies in a
	partial $K_2$ with its other endpoint $v$ outside the bag, label $C$ is
	sufficient. Indeed, all selected neighbours of $v$ already lie in
	$V(G_t)$. If $v$ were the centre of a larger global star, that star would
	therefore already be present in the processed part. Consequently, this
	$K_2$ either remains a $K_2$ globally or can grow only at $x$.
	More generally, the unique centre of a global star with at least two edges
	cannot be forgotten while one of its leaves remains outside the processed
	graph. When its centre is forgotten and a leaf remains in the bag, the
	partial component already has at least two edges and receives label $L$.
	
	At an introduce node, a newly active vertex is pending, is joined to a pending
	vertex to form a $K_2$, or is attached to a known centre. These are exactly the
	listed transitions. At a forget node, a pending vertex cannot occur in the
	restriction of $D$, because it has no selected neighbour in the processed
	graph and no neighbour outside it. The other forget transitions preserve
	the canonical label and count the edge between two active vertices, if
	present.
	
	At a join node, the two child partial solutions are restrictions of the same
	global packing and therefore activate the same bag vertices. For a singleton,
	a leaf with label $L$ in one child has no selected neighbour in the other
	child and hence meets label $P$ there. Otherwise the labels belong to
	$\{P,C\}$, whose merges are all listed. For a boundary edge $xy$, all
	counted attachments lie at the same endpoint, so its pair labels also
	combine by a valid merge. Induction now gives a DP computation representing
	$D$ and counting each of its edges exactly once.
\end{proof}

Lemmas~\ref{lem:chordal-sound} and~\ref{lem:chordal-complete} imply that
$\operatorname{DP}[r,\varnothing]=\rhooe(G)$.

\subsection{Running time and reconstruction}

Let $b_t=|B_t|$. A table contains one empty state, $3b_t$ singleton states,
and $3\binom{b_t}{2}$ pair states, hence $O((b_t+1)^2)$ states. Leaf,
introduce, and forget nodes can be processed in $O((b_t+1)^2)$ time.

At a join node, states are grouped by their active subset. There is one label
for the empty subset, three labels for each singleton subset, and three labels
for each two-element subset. Thus, only constantly many label pairs are
examined for each active subset, and a join node also takes
$O((b_t+1)^2)$ time.
The nice decomposition has $O(n^2)$ nodes and $b_t\leq n$, so the total
running time is $O(n^4)$.

For a disconnected graph with component orders $n_1,\ldots,n_q$, the total
running time is $O(\sum_i n_i^4)\subseteq O(n^4)$. To reconstruct an optimum
packing, store a predecessor transition for every improving DP update and
backtrack from the root. This does not change the asymptotic time bound.

\begin{theorem}
	\label{thm:chordal-algorithm}
	\label{thm:chordal-correct}
	For a chordal graph $G$ on $n$ vertices, the dynamic program computes
	$\rhooe(G)$ and an optimum edge open packing in $O(n^4)$ time.
\end{theorem}

\section{Maximum Edge Open Packing in Convex Bipartite Graphs}
\label{sec:convex}

Let $G=(X\cup Y,E)$ be a convex bipartite graph. Isolated vertices do not
affect an edge open packing and are deleted first; if no edge remains, the
optimum is zero. Let $X=\{x_1,x_2,\ldots,x_p\}$ be a convex ordering. For
each $y\in Y$, write
$N(y)=\{x_{\ell(y)},x_{\ell(y)+1},\ldots,x_{r(y)}\}$.

By Observation~\ref{obs:star-characterization}, the components of
$G[V(D)]$ for an edge open packing $D$ are stars $K_{1,t}$ with $t\geq1$.
In particular, $K_2$ is included. Its centre may be chosen in either part;
when describing an existing packing, we choose its endpoint in $X$.
Since $X$ and $Y$ are independent sets, a collection of vertex-disjoint
stars forms an edge open packing exactly when there is no edge between
two distinct stars.

\subsection{Separation properties}

For a component $S$ of $G[V(D)]$, define
$x_{\min}(S)=\min\{i:x_i\in V(S)\cap X\}$. Every component contains an
edge, so this minimum is defined. List the components as $S_1,\ldots,S_k$
in increasing order of $x_{\min}$; their minima are distinct because their
vertex sets are disjoint. A \emph{prefix} is
$\mathcal S_j=\{S_1,\ldots,S_j\}$ for some $0\leq j\leq k$, with
$\mathcal S_0=\varnothing$. A component $S_t$ is \emph{later} than this
prefix if $t>j$.

\begin{claim}
\label{clm:convex-order}
If $a<b$, every $X$-vertex of $S_a$ has index smaller than
$x_{\min}(S_b)$.
\end{claim}

\begin{proof}
Put $i=x_{\min}(S_b)$. An $X$-centred star contains just one $X$-vertex.
Otherwise, if $S_a$ contained $X$-vertices on both sides of index $i$,
convexity would make its $Y$-centre adjacent to $x_i$, joining two distinct
components. Equality of indices is excluded by vertex-disjointness.
\end{proof}

For a prefix $\mathcal S$, define
$\alpha(\mathcal S)=\max\{i:x_i\in V(S)\cap X\text{ for some }S\in\mathcal S\}$
and
$\beta(\mathcal S)=\max\{r(y):y\in V(S)\cap Y\text{ for some }S\in\mathcal S\}$,
taking both maxima to be zero for the empty prefix. We always have
$\alpha(\mathcal S)\leq\beta(\mathcal S)$: each selected $X$-vertex has a
selected $Y$-neighbour whose right endpoint is at least its index.

\begin{lemma}
\label{lem:convex-separation}
Let $\mathcal S$ be a prefix, and put $\alpha=\alpha(\mathcal S)$ and
$\beta=\beta(\mathcal S)$. In every later component, each $X$-vertex
$x_i$ satisfies $i>\beta$, and each $Y$-vertex $y$ satisfies
$\ell(y)>\alpha$.
\end{lemma}

\begin{proof}
The empty-prefix case is immediate. Suppose that a later selected vertex
$x_i$ satisfies $i\leq\beta$. A previously selected $Y$-vertex $y$ has
$r(y)=\beta$ and a selected neighbour $x_h$ in its component.
Claim~\ref{clm:convex-order} gives $h<i$, so
$\ell(y)\leq h<i\leq r(y)$. Convexity implies $yx_i\in E(G)$, a
contradiction. Thus $i>\beta\geq\alpha$.

If a later selected $Y$-vertex $z$ had $\ell(z)\leq\alpha$, its selected
neighbour $x_i$ would satisfy $\ell(z)\leq\alpha<i\leq r(z)$. It would
therefore be adjacent to the previously selected vertex $x_\alpha$,
again a contradiction.
\end{proof}

\subsection{Dynamic programming table}

A \emph{labelled partial solution} consists of an edge open packing $D'$
together with barriers $(\alpha,\beta)$. Its components are generated from
left to right by the transitions below. The value $\alpha$ is the exact
index of its rightmost selected $X$-vertex, with $\alpha=0$ only for the
empty solution. The recorded barrier $\beta$ is an upper bound on $r(y)$
for every selected $y\in Y$. It may exceed their actual maximum; such a
conservative barrier restricts the permitted continuation. In every
reachable state $\alpha\leq\beta$. The continuation rule requires every
future selected $X$-vertex to have index greater than $\beta$ and every
future selected $Y$-vertex $y$ to satisfy $\ell(y)>\alpha$.

For $0\leq\alpha,\beta\leq p$, let $F[\alpha,\beta]$ be the maximum
number of edges in a reachable labelled partial solution with these
barriers, or $-\infty$ if none exists. Future choices depend only on the
barriers, so retaining the largest value for each pair is sufficient.
Initialize $F[0,0]=0$ and all other entries to $-\infty$. Each transition
strictly increases $\beta$: its new value is either $q\geq i>\beta$ or
$r(y)\geq q\geq s>\beta$. The states can therefore be processed in
increasing order of $\beta$.

\subsection{Transitions}

We append one new star component to a reachable partial solution.

\medskip
\noindent
\textbf{$X$-centred star.}
Choose a centre $x_i$ with $i>\beta$ and a new right-endpoint barrier
$q\in\{i,\ldots,p\}$. Define
\[
w_X(i,\alpha,q)
=\bigl|\{y\in Y:\alpha<\ell(y)\leq i\leq r(y)\leq q\}\bigr|.
\]
If $w_X(i,\alpha,q)>0$, select all these vertices as leaves and update
\[
F[i,q]\gets\max\{F[i,q],F[\alpha,\beta]+w_X(i,\alpha,q)\}.
\]
The condition $i>\beta$ excludes edges from the new centre to previously
selected $Y$-vertices, and $\ell(y)>\alpha$ excludes edges from new
leaves to previously selected $X$-vertices. The updated barriers require
future $X$-vertices to lie beyond $q$ and future $Y$-vertices to have left
endpoint greater than $i$; hence the new star is also separated from all
permitted future components. The recorded value $q$ need not be attained
by a selected leaf, which is why conservative barriers are allowed.

\medskip
\noindent
\textbf{$Y$-centred star.}
Choose a centre $y\in Y$ with $\ell(y)>\alpha$, and put
$s=\max\{\ell(y),\beta+1\}$. If $s\leq r(y)$, then for every
$q\in\{s,\ldots,r(y)\}$ select $x_s,\ldots,x_q$ as leaves and update
\[
F[q,r(y)]\gets\max\{F[q,r(y)],F[\alpha,\beta]+q-s+1\}.
\]
The new leaves have indices greater than $\beta$, so they are
nonadjacent to previously selected $Y$-vertices. The condition
$\ell(y)>\alpha$ separates the new centre from previously selected
$X$-vertices. Future $Y$-vertices have left endpoint greater than $q$,
and future $X$-vertices have index greater than $r(y)$, so neither can
have an edge to the new star.

The algorithm returns $\max_{0\leq\alpha,\beta\leq p}F[\alpha,\beta]$.

\subsection{Correctness}

\begin{lemma}
\label{lem:convex-sound}
Every solution constructed by the dynamic program is an edge open packing.
\end{lemma}

\begin{proof}
Both transitions add a star whose vertices have no edges to earlier
selected vertices, as shown above, and preserve the stated barrier
invariants. In particular, no previously selected vertex can be reused:
a previous $X$-vertex has index at most $\alpha\leq\beta$, while a
previous $Y$-vertex has a selected neighbour of index at most $\alpha$
and therefore has left endpoint at most $\alpha$.
Inductively, the selected vertices induce disjoint stars, and
Observation~\ref{obs:star-characterization} applies.
\end{proof}

\begin{lemma}
\label{lem:convex-complete}
For every edge open packing $D$, the dynamic program constructs a solution
of cardinality at least $|D|$.
\end{lemma}

\begin{proof}
Process $S_1,\ldots,S_k$ in their defined order, possibly adding leaves.
We maintain the exact barriers of the corresponding original prefix and
keep the enlarged prefix compatible with the unprocessed components.
Initially both barriers are zero. Suppose the current barriers are
$\alpha,\beta$, and $S$ is the next component.
Lemma~\ref{lem:convex-separation} gives index greater than $\beta$ for
every $X$-vertex of $S$, and left endpoint greater than $\alpha$ for
every $Y$-vertex of $S$.

If $S$ is centred at $x_i\in X$, let
$q=\max\{r(y):y\in V(S)\cap Y\}$. Every original leaf is counted by
$w_X(i,\alpha,q)$, so the corresponding transition selects at least
$|E(S)|$ edges. An added leaf has left endpoint greater than $\alpha$
and right endpoint at most $q$, excluding edges to all earlier selected
$X$-vertices. Applying Lemma~\ref{lem:convex-separation} to the original
prefix including $S$ shows that every later selected $X$-vertex has
index greater than $q$, excluding edges to these vertices as well.
An added leaf cannot already belong to a later component, since its edge
to $x_i$ would join that component to $S$. Since an original leaf attains
$q$ and $q\geq i>\beta$, the enlarged prefix has exact barriers $(i,q)$,
the same as the original prefix including $S$.

If $S$ is centred at $y\in Y$, put
$q=\max\{h:x_h\in V(S)\cap X\}$ and
$s=\max\{\ell(y),\beta+1\}$. All original leaves lie in
$x_s,\ldots,x_q$, so the transition again selects at least $|E(S)|$
edges. Any added leaf has index greater than $\beta$, excluding edges
to previously selected $Y$-vertices. Applying
Lemma~\ref{lem:convex-separation} to the original prefix including $S$
shows that every later selected $Y$-vertex $z$ satisfies $\ell(z)>q$.
Thus none of the added leaves has an edge to a later selected
$Y$-vertex. Claim~\ref{clm:convex-order} also ensures that none is a
selected $X$-vertex of a later component. The exact barriers are
$(q,r(y))$, unchanged from those of the original prefix including $S$,
since $x_q$ is an original leaf and $r(y)\geq q>\beta$.

Both transitions therefore preserve the induction invariant. Their
sequence gives a reachable value at least $|D|$.
\end{proof}

\subsection{Running time and reconstruction}

Deleting isolated vertices and extracting interval endpoints from a given
convex ordering takes $O(n+m)$ time. If an ordering is not supplied, it
can be found in $O(n+m)$ time by consecutive-ones testing
\cite{BoothLueker1976}.

The table has $(p+1)^2=O(p^2)$ entries. Compute
$C[a,b]=|\{y\in Y:\ell(y)=a,\ r(y)=b\}|$ for $1\leq a\leq b\leq p$,
with zero entries elsewhere. A two-dimensional prefix-sum table gives
$w_X(i,\alpha,q)=\sum_{a=\alpha+1}^{i}\sum_{b=i}^{q}C[a,b]$ in $O(1)$
time. This preprocessing takes $O(p^2+|Y|)$ time and space.
From one state there are $O(p^2)$ $X$-centred transitions and
$O(|Y|p)$ $Y$-centred transitions. The total time is therefore
$O(p^4+p^3|Y|)=O(n^4)$.

For each improving update, store its predecessor state, transition type,
centre, and endpoint $q$. Backtracking identifies at most $p$ transitions,
since $\beta$ strictly increases. Recover the leaves of each $X$-centred
star by scanning $Y$, and those of each $Y$-centred star by listing its
selected interval. Reconstruction takes $O(p|Y|+p^2)$ time. The tables,
interval endpoints, and predecessor records use $O(p^2+|Y|)$ auxiliary
space once the ordering is available; including the input and ordering
preprocessing, $O(n^2)$ space suffices.

\begin{theorem}
\label{thm:convex-correct}
\label{thm:convex-algorithm}
The \MEOP\ problem can be solved in $O(n^4)$ time and $O(n^2)$ space on
convex bipartite graphs.
\end{theorem}

\begin{proof}
Lemma~\ref{lem:convex-sound} bounds the returned value from above by
$\rhooe(G)$, and Lemma~\ref{lem:convex-complete} gives the reverse
inequality. The preceding analysis establishes the time and space bounds
and the reconstruction of an optimum packing.
\end{proof}

\section{Conclusion}
\label{sec:conclusion}

We proved that the oriented star-conflict graph of Das and
Santra~\cite{DasSantra2026} preserves AT-freeness and used this property
to obtain an $O(n^2+m^4)$-time algorithm for \MEOP\ on AT-free graphs.
We also gave $O(n^4)$-time dynamic programs for chordal and convex
bipartite graphs. The former handles unbounded clique number, while the
latter extends the earlier biconvex algorithm.

Natural next questions are whether these running times can be improved
and whether \MEOP\ admits a polynomial-time algorithm on chordal
bipartite graphs. Distance-hereditary graphs are already covered by
Bhyravarapu et al.~\cite{BhyravarapuEtAl2026}. It is also useful to study
further classes for which the construction $G\mapsto A_G$ preserves a
structure that makes maximum independent set tractable.

For general graphs, fixed-parameter tractability by treewidth follows
from monadic second-order optimisation~\cite{CourcelleMosbah1993}:
feasibility of a chosen edge set is expressed by excluding two selected
edges together with a third edge joining their endpoints, and the
objective is its cardinality. Thus, the questions concern explicit
algorithms, their dependence on treewidth, and kernelisation bounds.
An independent set $S$ is asteroidal if $S\setminus\{v\}$ lies in one
component of $G-N_G[v]$ for every $v\in S$; its maximum size is the
asteroidal number. The preservation proof maps an asteroidal set of $A_G$ to an
asteroidal set of $G$ by taking the leaf coordinates: the same path
argument applies to any two members while avoiding a third. Combined
with the bounded-asteroidal-number result of Broersma
et al.~\cite{BroersmaEtAl1999}, this gives polynomial-time solvability
for each fixed bound on asteroidal number. Whether the dependence can
be made fixed-parameter tractable is a more precise further question.

\bibliographystyle{plain}% the mandatory bibstyle
\bibliography{EOP_ref}

\end{document}